\documentclass[letterpaper,10pt,conference]{ieeeconf}

\IEEEoverridecommandlockouts  

\usepackage{amsmath}
\usepackage{amssymb}
\usepackage{amsthm}
\usepackage{mathtools}
\usepackage{derivative}
\usepackage{xcolor}
\usepackage{bm}

\usepackage[noadjust]{cite}
\usepackage{url}

\usepackage{graphicx}
\usepackage[export]{adjustbox} 
\graphicspath{{figures/}}

\newtheorem{theorem}{Theorem}

\newtheorem{proposition}{Proposition}
\newtheorem{corollary}{Corollary}

\theoremstyle{definition}
\newtheorem{definition}{Definition}

\newtheorem{remark}{Remark}

\newtheorem{example}{Example}

\newcommand{\rank}{\operatorname{rank}}

\renewcommand{\bf}{\mathbf{f}} 

\usepackage{xcolor}
\definecolor{myblue}{RGB}{49, 114, 174}
\definecolor{myred}{rgb}{0.796, 0.235, 0.2}
\definecolor{mygreen}{rgb}{0.22, 0.596, 0.149}
\definecolor{mypurple}{rgb}{0.584,0.345,0.698}

\usepackage{blindtext}

\usepackage{comment}

\usepackage{cite}
\usepackage{amsfonts}

\usepackage{algorithmic}
\usepackage{textcomp}

\title{\textbf{
Stability Analysis in Multi-Constraint  Safety Filters for Linear Systems
}}

\author{Shima Sadat Mousavi, Pol Mestres, and Aaron D. Ames %
\thanks{The authors are with the Department of Mechanical and Civil Engineering, California Institute of Technology, Pasadena, CA \texttt{\{smousavi,mestres,ames\}@caltech.edu}.}%
\thanks{This research was supported by the Boeing Strategic University Initiative.}%
}

\begin{document}

\bstctlcite{IEEEexample:BSTcontrol}

\maketitle
\thispagestyle{empty}
\pagestyle{empty}


\begin{abstract}
Multi-constraint safety filters based on control barrier functions for linear systems with affine state constraints yield continuous piecewise-affine closed-loop dynamics and may introduce boundary equilibria and unstable active-set modes. Although these filters guarantee forward invariance, they can change nominal stability, and it remains unclear when unstable modes cause divergence versus bounded, convergent behavior. This paper develops a geometric framework to analyze these cases: leveraging explicit active-set realizations, we show that equilibria associated with nonempty active sets lie on the corresponding constraint faces and that any unstable directions are tangent to those faces due to exponential enforcement of the active constraints. We characterize mode stability via a minimum-phase test, provide a sufficient recession-cone certificate for divergence under a fixed active set, and derive tractable linear-matrix-inequality conditions for global exponential stability or boundedness using Lyapunov and LaSalle arguments.
\end{abstract}

\section{Introduction}
Safety-critical control requires feedback laws that achieve performance objectives while enforcing hard constraints, typically formalized as forward invariance of a safe set. Control barrier functions (CBFs) enable a practical ``safety filter'' architecture: a nominal controller is minimally modified to enforce safety \cite{ames2014controlbarrier,ames2017cbfqp,xu2015robustcbf}, with extensions to output-feedback and higher-relative-degree settings \cite{mousavi2026functional}. This separation of nominal performance design from safety enforcement is attractive in practice, but it also complicates analysis: the safety filter can alter stability and asymptotic behavior relative to the nominal closed loop, especially when multiple constraints are enforced simultaneously. This paper studies the following question for linear systems with affine state constraints: when can multi-constraint CBF-QP safety filters introduce boundary equilibria, unstable active-set modes, or unbounded trajectories, and when do trajectories remain bounded and converge to the desired equilibrium?

A growing literature has documented that CBF-QP safety filters may exhibit undesired dynamical behaviors beyond constraint satisfaction, including spurious equilibria, limit cycles, and divergence \cite{XT-DVD:21,MFR-APA-PT:21,PM-YC-EDA-JC:25-jnls}, but it offers limited structural explanation of when local instability leads to unbounded trajectories versus bounded long-term behavior.
This gap is particularly evident for linear systems with affine constraints, which provide a tractable setting for isolating how safety filtering alters nominal closed-loop behavior. For the single-constraint case, a sharp characterization is available because the filtered closed loop reduces to a two-mode piecewise-affine system \cite{PM-SSM-ADA:26,mousavi2026stability}; complementary work characterizes safe stabilizing linear feedback directly \cite{mestres2026safestabilizing}. These results do not extend directly to multi-constraint safety filters, where many active sets can arise and trajectories may switch across faces of the safe polyhedron. Related multi-constraint works address feasibility and compatibility of simultaneous CBF inequalities \cite{mousavi2025vertices,mousavi2026structure}, but not stability, equilibria, or divergence of the resulting closed-loop system.

In the multi-constraint regime, explicit and parametric QP results \cite{bemporad2002explicit,tondel2003explicit,borrelli2017mpc} imply a continuous piecewise-affine closed loop under standard regularity assumptions. In particular, for linear dynamics with affine constraints, the optimizer is piecewise affine and the state space admits a polyhedral active-set partition. General stability tools for piecewise-affine and switched systems provide useful sufficient conditions \cite{branicky1998multiple,liberzon2003switching,pavlov2007convergence}, but they do not provide a tailored dynamical understanding of multi-constraint CBF-QP safety filters, even though such filters are often used as modular wrappers and can reshape nominal stability and equilibria. This gap matters for reliable deployment because invariance alone does not preclude instability, unbounded behavior, or failure to converge to the intended operating point.

This paper fills this gap via a geometric and invariant-set analysis that leverages explicit active-set realizations of multi-constraint CBF-QPs \cite{mestres2025explicit}. We first provide a spectral characterization of active-set modes and establish a minimum-phase interpretation of mode stability (Thm.~\ref{thm:det_factorization_multi}, Cor.~\ref{cor:eigs_minphase_multi}). We then characterize the equilibrium structure of the filtered closed loop, showing that equilibria associated with nonempty active sets lie on constraint boundaries (Prop.~\ref{prop:equilibria_on_faces}) and that instability under a fixed active set can only develop tangentially to the active constraints (Prop.~\ref{prop:tangency_active_modes}, Cor.~\ref{cor:tangency_unstable_directions}). Next, we provide a sufficient recession-cone certificate for unbounded trajectories that persist under a fixed active set (Prop.~\ref{prop:persistent_unstable_ray}). Finally, we derive tractable sufficient Lyapunov/LaSalle linear-matrix-inequality (LMI) conditions that certify global exponential stability or boundedness on the safe set (Thm.~\ref{thm:ges_cqlf}, Prop.~\ref{prop:lmi_lasalle}).

\section{Problem Setup and Closed-Loop Structure}

\textbf{Notation.}
The sets of real numbers and $n$-dimensional Euclidean space are denoted by $\mathbb R$ and $\mathbb R^n$.
We write $[p]:=\{1,\dots,p\}$.
The interior and closure of a set $\mathcal S\subset\mathbb R^n$ are denoted by $\operatorname{int}(\mathcal S)$ and $\operatorname{cl}(\mathcal S)$.
For vectors $x,y\in\mathbb R^n$, the inequalities $x\ge y$ and $x>y$ are understood componentwise.
For symmetric matrices, positive definiteness and semidefiniteness are denoted by $P\succ 0$ and $P\succeq 0$.
The identity matrix of dimension $n$ is denoted by $I_n$, and by $I$ when the dimension is clear from context. For a matrix $G\succ0$, the weighted norm is defined as
$\|z\|_G^2 := z^\top G z$. For scalars $a_1,\dots,a_k$, $\mathrm{diag}(a_1,\dots,a_k)$ denotes the diagonal matrix with diagonal entries $a_1,\dots,a_k$.

\subsection{Linear System, Safety Constraints, and Safety Filter}

We consider the linear control system
\begin{equation}
\dot{x} = A x + B u,
\label{eq:system}
\end{equation}
where $x \in \mathbb{R}^n$ is the state and $u \in \mathbb{R}^m$ is the control
input. A nominal state-feedback controller of the form
\begin{equation}
u = -K x
\label{eq:nominal}
\end{equation}
is assumed to be given, with $A - B K$ Hurwitz.

Safety is specified by $p$ affine state constraints
\begin{equation}
h_i(x) := c_i^\top x + d_i \ge 0,
\quad i = 1,\dots,p,
\label{eq:constraints}
\end{equation}
where $c_i \in \mathbb{R}^n$ and $d_i > 0$. The corresponding safe set is
\begin{equation}
\mathcal{C} :=
\bigcap_{i=1}^p \{ x \in \mathbb{R}^n : h_i(x) \ge 0 \}.
\label{eq:safeset}
\end{equation}
The condition $d_i > 0$ for all $i\in[p]$ ensures that the origin lies in the interior of
$\mathcal{C}$. Throughout the paper, we restrict attention to first-order
(relative-degree-one) safety constraints, i.e., $c_i^\top B\neq0$ for all $i\in[p]$.

To enforce safety, we employ a CBF-based safety filter defined as the solution of the following QP:
\begin{equation}
\begin{aligned}
u^\star(x) = \arg\min_{u \in \mathbb{R}^m} \;&
\frac{1}{2}\|u + K x\|_G^2 \\
\text{s.t.} \;&
c_i^\top (A x + B u)
\ge -\alpha (c_i^\top x + d_i),
\\
& i = 1,\dots,p,
\end{aligned}
\label{eq:cbf_qp}
\end{equation}
where $G \in \mathbb{R}^{m \times m}$ is symmetric positive definite and
$\alpha > 0$ is a fixed scalar parameter.

Throughout the paper, we assume that the quadratic program \eqref{eq:cbf_qp}
is feasible for all $x \in \mathcal{C}$. Since $G\succ0$, its optimizer is unique; the feasibility assumption excludes configurations in which the simultaneous CBF inequalities are incompatible. Under these assumptions, the safety filter renders $\mathcal{C}$ forward invariant and induces continuous piecewise-affine closed-loop dynamics.

\subsection{Piecewise-Affine Closed-Loop Representation}
\label{subsec:pwa_representation}

We now describe the structure of the closed-loop dynamics induced by the
safety filter \eqref{eq:cbf_qp}. Although the control input is defined
implicitly through a quadratic program, the resulting closed-loop system
admits an explicit representation as a continuous piecewise-affine system
on the safe set $\mathcal{C}$, which will be the basis for all subsequent
analysis.

For any $x \in \mathcal{C}$, let $\mathcal{I}(x) \subseteq [p]$
denote the set of indices of constraints that are active at the optimizer of
\eqref{eq:cbf_qp}. Since $G\succ0$, \eqref{eq:cbf_qp} has a unique optimizer $u^\star(x)$ for each $x\in\mathcal C$.
Assume the active constraints have linearly independent $u$-gradients, i.e.,
$\rank(B^\top C_{\mathcal I(x)})=|\mathcal I(x)|$ for all $x\in\mathcal C$.
This standard LICQ assumption excludes degenerate active sets and enables the explicit active-set representation below.
This induces a partition of $\mathcal{C}$ into regions
indexed by active sets,
\begin{equation}
\mathcal{R}_{\mathcal{I}}
:= \{ x \in \mathcal{C} : \mathcal{I}(x) = \mathcal{I} \},
\label{eq:regions}
\end{equation}
where $\mathcal{I}$ ranges over all active sets that occur on $\mathcal{C}$.
Under standard regularity conditions for strictly convex quadratic programs
with affine constraints, these regions are polyhedral and form a finite
partition of $\mathcal{C}$. For each active set $\mathcal I$, define the closed cell
\(
\overline{\mathcal R}_{\mathcal I}:=\operatorname{cl}(\mathcal R_{\mathcal I}).
\)
In the affine-data setting, $\overline{\mathcal R}_{\mathcal I}$ is obtained by
replacing strict inequalities in the KKT region characterization by non-strict
ones.

On each region $\mathcal{R}_{\mathcal{I}}$, the optimizer $u^\star(x)$ is
affine in $x$, and the closed-loop dynamics take the affine form
\begin{equation}
\dot{x} = \tilde{A}_{\mathcal{I}} x + \tilde{b}_{\mathcal{I}},
\qquad x \in \mathcal{R}_{\mathcal{I}},
\label{eq:pwa_dynamics}
\end{equation}
for matrices $\tilde{A}_{\mathcal{I}} \in \mathbb{R}^{n\times n}$ and vectors
$\tilde{b}_{\mathcal{I}} \in \mathbb{R}^n$ that depend on the active set
$\mathcal{I}$. When no constraint is active,
$\mathcal{I} = \emptyset$, the safety filter coincides with the nominal
controller, and
\begin{equation}
A_0 :=\tilde{A}_{\emptyset} = A - B K,
\qquad
\tilde{b}_{\emptyset} = 0.
\label{eq:nominal_mode}
\end{equation}
We refer to the affine dynamics associated with any fixed active set
$\mathcal{I}$ as a \emph{mode} of the closed-loop system.

We next give explicit expressions for the matrices in \eqref{eq:pwa_dynamics}. For each index set $\mathcal I\subseteq[p]$, let $C_{\mathcal I}$ denote the matrix whose columns are the vectors $\{c_i\}_{i\in\mathcal I}$, and let $d_{\mathcal I}$ denote the vector collecting the corresponding scalars $\{d_i\}_{i\in\mathcal I}$:
\begin{equation}
C_{\mathcal I}:=\begin{bmatrix} c_i \end{bmatrix}_{i\in\mathcal I}
\in\mathbb R^{n\times|\mathcal I|},
\qquad
d_{\mathcal I}:=\begin{bmatrix} d_i \end{bmatrix}_{i\in\mathcal I}
\in\mathbb R^{|\mathcal I|}.
\label{eq:CI_dI_def}
\end{equation}
Let
\begin{equation}
S_{\mathcal{I}}
:= C_{\mathcal{I}}^\top B G^{-1} B^\top C_{\mathcal{I}}
\in \mathbb{R}^{|\mathcal{I}|\times|\mathcal{I}|}.
\label{eq:S_I}
\end{equation}
For every active set $\mathcal I$ that occurs, $S_{\mathcal I}\succ0$ by the LICQ assumption above. Define
\begin{equation}
\Pi_{\mathcal{I}}
:= B G^{-1} B^\top C_{\mathcal{I}}\, S_{\mathcal{I}}^{-1}
\in \mathbb{R}^{n\times|\mathcal{I}|}.
\label{eq:P_I}
\end{equation}
Then, on $\mathcal{R}_{\mathcal{I}}$, the closed-loop matrices are given by
\begin{align}
\tilde A_{\mathcal I}&=A_0-\Pi_{\mathcal I}C_{\mathcal I}^\top(A_0+\alpha I),
\label{eq:Atilde_I_explicit}\\
\tilde{b}_{\mathcal{I}}
&= -\alpha\, \Pi_{\mathcal{I}}\, d_{\mathcal{I}}.
\label{eq:btilde_I_explicit}
\end{align}

\begin{remark}
Equations \eqref{eq:Atilde_I_explicit}--\eqref{eq:btilde_I_explicit} follow
from the KKT conditions of \eqref{eq:cbf_qp} by treating the constraints in
$\mathcal{I}$ as equalities and solving for the corresponding Lagrange
multipliers; see, e.g.,~\cite{mestres2025explicit}.
\end{remark}

\subsection{Spectral Structure of Active-Set Modes}
\label{subsec:spectral_structure}

This subsection characterizes the spectrum of each affine mode matrix
$\tilde A_{\mathcal I}$ in \eqref{eq:Atilde_I_explicit} in terms of the nominal
closed-loop matrix $A_0:=A-BK$ and a square transfer matrix associated with the
active constraints. This result will be used to interpret mode stability as a
minimum-phase property. Define
\begin{equation}\label{eq:F_I_def}
F_{\mathcal I}(\lambda)
:= C_{\mathcal I}^\top(\lambda I-A_0)^{-1}BG^{-1}B^\top C_{\mathcal I}.
\end{equation}

\begin{theorem}
\label{thm:det_factorization_multi}
Let $\mathcal I\neq\emptyset$ and suppose $S_{\mathcal I}$ in \eqref{eq:S_I} is nonsingular.
For any $\lambda\in\mathbb C$ such that $\lambda I-A_0$ is invertible,
\begin{equation}\label{eq:det_factorization_multi}
\det(\lambda I-\tilde A_{\mathcal I})
=
\det(\lambda I-A_0)\;
\frac{(\lambda+\alpha)^{|\mathcal I|}}{\det(S_{\mathcal I})}\;
\det\!\big(F_{\mathcal I}(\lambda)\big).
\end{equation}
\end{theorem}

\begin{proof}
Write $\tilde A_{\mathcal I}=A_0-U_{\mathcal I}V_{\mathcal I}^\top$ with
$U_{\mathcal I}:=\Pi_{\mathcal I}$ and $V_{\mathcal I}^\top:=C_{\mathcal I}^\top(A_0+\alpha I)$.
By the matrix determinant lemma,
\[
\det(\lambda I-\tilde A_{\mathcal I})
=\det(\lambda I-A_0)\det\!\Big(I+V_{\mathcal I}^\top(\lambda I-A_0)^{-1}U_{\mathcal I}\Big).
\]
Using $(A_0+\alpha I)(\lambda I-A_0)^{-1}=-I+(\lambda+\alpha)(\lambda I-A_0)^{-1}$
and $C_{\mathcal I}^\top \Pi_{\mathcal I}=I_{|\mathcal I|}$ yields
\[
I+V_{\mathcal I}^\top(\lambda I-A_0)^{-1}U_{\mathcal I}
=(\lambda+\alpha)\,C_{\mathcal I}^\top(\lambda I-A_0)^{-1}\Pi_{\mathcal I}.
\]
Substituting $\Pi_{\mathcal I}=BG^{-1}B^\top C_{\mathcal I}S_{\mathcal I}^{-1}$
gives \eqref{eq:det_factorization_multi}.
\end{proof}

\begin{corollary}
\label{cor:eigs_minphase_multi}
Fix $\mathcal I\neq\emptyset$ and assume $S_{\mathcal I}$ in \eqref{eq:S_I} is nonsingular.
Define
\[
M:=BG^{-1}B^\top,
\qquad
\mathcal Z_{\mathcal I}(\lambda)
:=
\begin{bmatrix}
\lambda I-A_0 & -M C_{\mathcal I}\\
C_{\mathcal I}^\top & 0
\end{bmatrix}.
\]
Then, for all $\lambda\in\mathbb C$,
\begin{equation}
\det(\lambda I-\tilde A_{\mathcal I})
=
\frac{(\lambda+\alpha)^{|\mathcal I|}}{\det(S_{\mathcal I})}\,
\det\!\big(\mathcal Z_{\mathcal I}(\lambda)\big).
\label{eq:det_rosenbrock_multi}
\end{equation}
Consequently:
\begin{enumerate}
\item $\lambda=-\alpha$ is an eigenvalue of $\tilde A_{\mathcal I}$ with algebraic multiplicity at least $|\mathcal I|$;
\item for every $\lambda\neq -\alpha$, $\lambda$ is an eigenvalue of $\tilde A_{\mathcal I}$ if and only if $\det(\mathcal Z_{\mathcal I}(\lambda))=0$.
\end{enumerate}
Equivalently, every eigenvalue of $\tilde A_{\mathcal I}$ different from $-\alpha$
coincides with a finite invariant zero of the square system
\[
\Sigma_{\mathcal I}:\quad
\dot x = A_0x+M C_{\mathcal I}v,
\qquad
y=C_{\mathcal I}^\top x.
\]
In particular, $\tilde A_{\mathcal I}$ is Hurwitz if and only if
$\Sigma_{\mathcal I}$ is minimum phase, i.e., all finite invariant zeros of
$\Sigma_{\mathcal I}$ lie in $\{\lambda\in\mathbb C:\Re(\lambda)<0\}$.
\end{corollary}

\begin{proof}
By the Schur complement formula,
\(
\det(\mathcal Z_{\mathcal I}(\lambda))
=
\det(\lambda I-A_0)\det(F_{\mathcal I}(\lambda))
\)
for every $\lambda$ such that $\lambda I-A_0$ is invertible. Combining this with \eqref{eq:det_factorization_multi} gives
\[
\det(\lambda I-\tilde A_{\mathcal I})
=
\frac{(\lambda+\alpha)^{|\mathcal I|}}{\det(S_{\mathcal I})}\,
\det(\mathcal Z_{\mathcal I}(\lambda)).
\]
Since both sides are polynomials in $\lambda$, this identity holds for all $\lambda\in\mathbb C$. The remaining claims follow immediately.
\end{proof}

\begin{remark}
The factorization \eqref{eq:det_factorization_multi} is stated for
$\lambda I-A_0$ invertible, whereas Cor.~\ref{cor:eigs_minphase_multi} avoids this restriction.
It also separates the $|\mathcal I|$ modes at $-\alpha$, associated with contraction normal to the active faces, from the remaining invariant-zero dynamics, whose unstable directions are tangent to those faces by Cor.~\ref{cor:tangency_unstable_directions}. Thus, $\alpha$ sets the normal contraction rate but does not alter the tangential invariant zeros.
\end{remark}

\section{Equilibria and Active-Face Geometry}
\label{sec:equilibria_geometry}

In this section, we study the equilibrium structure of the piecewise-affine closed-loop dynamics \eqref{eq:pwa_dynamics}. Since each active set defines a distinct affine mode, multiple constraints can generate several isolated equilibria or even continua on intersections of constraint boundaries.

A point $x^\star\in\mathcal C$ is an equilibrium of the closed-loop dynamics \eqref{eq:pwa_dynamics} if and only if there exists an index set $\mathcal I$ such that
\begin{equation}
x^\star\in\overline{\mathcal R}_{\mathcal I}
\quad\text{and}\quad
\tilde A_{\mathcal I}x^\star+\tilde b_{\mathcal I}=0;
\label{eq:eq_characterization}
\end{equation}
see, e.g., \cite{bemporad2000verification,johansson2003piecewise}. In particular, if $\tilde A_{\mathcal I}$ is nonsingular, then the corresponding candidate equilibrium is
\[
x_{\mathcal I}:=-\tilde A_{\mathcal I}^{-1}\tilde b_{\mathcal I},
\]
which is an equilibrium of \eqref{eq:pwa_dynamics} only if $x_{\mathcal I}\in\overline{\mathcal R}_{\mathcal I}$. Thus, a continuum of equilibria associated with a fixed mode can occur only if $\tilde A_{\mathcal I}$ is singular.

Throughout the paper, we take the desired equilibrium of the nominal closed-loop system as the reference point and, without loss of generality, place it at the origin by a translation of coordinates. All subsequent statements are made in these coordinates.

\subsection{Active-Face Structure of Equilibria}
\label{subsec:active_face_equilibria}

We now establish a geometric property of equilibria induced by multi-constraint safety filters: any equilibrium associated with a nonempty active set lies on the boundary of the safe set, more precisely on the face determined by the active constraints. This property will play a central role in the analysis that follows. Recall that $\overline{\mathcal R}_{\mathcal I}$ denotes the closed cell associated with the active set $\mathcal I$.

\begin{proposition}
\label{prop:equilibria_on_faces}
Let $\mathcal I\neq\emptyset$. If $x^\star\in\overline{\mathcal R}_{\mathcal I}$ is an equilibrium of the closed-loop dynamics \eqref{eq:pwa_dynamics}, then
\begin{equation}
C_{\mathcal I}^\top x^\star+d_{\mathcal I}=0,
\label{eq:equilibrium_on_face}
\end{equation}
where $C_{\mathcal I}$ and $d_{\mathcal I}$ are defined in \eqref{eq:CI_dI_def}.
\end{proposition}

\begin{proof}
Fix a nonempty active set $\mathcal{I}$ and let $x^\star \in \overline{\mathcal R}_{\mathcal{I}}$
be a closed-loop equilibrium associated with this mode. Then
$\tilde{A}_{\mathcal{I}} x^\star + \tilde{b}_{\mathcal{I}} = 0$.
Substituting \eqref{eq:Atilde_I_explicit}--\eqref{eq:btilde_I_explicit} yields
\begin{align*}
0
&= (A-BK)x^\star
- \Pi_{\mathcal{I}}\, C_{\mathcal{I}}^\top (A-BK+\alpha I)x^\star
- \alpha \Pi_{\mathcal{I}} d_{\mathcal{I}} \\
&= (A-BK)x^\star
- \Pi_{\mathcal{I}}\Big(
C_{\mathcal{I}}^\top (A-BK+\alpha I)x^\star + \alpha d_{\mathcal{I}}
\Big).
\end{align*}
Premultiplying by $C_{\mathcal{I}}^\top$ gives
\begin{align*}
0
&= C_{\mathcal I}^\top (A-BK)x^\star \\
&\quad
- C_{\mathcal I}^\top \Pi_{\mathcal I}
\Big(
C_{\mathcal I}^\top (A-BK+\alpha I)x^\star + \alpha d_{\mathcal I}
\Big).
\end{align*}
Using \eqref{eq:P_I}--\eqref{eq:S_I}, we have
\[
C_{\mathcal{I}}^\top \Pi_{\mathcal{I}}
= C_{\mathcal{I}}^\top B G^{-1} B^\top C_{\mathcal{I}}\, S_{\mathcal{I}}^{-1}
= I_{|\mathcal{I}|}.
\]
Hence
\begin{align*}
0
&= C_{\mathcal{I}}^\top (A-BK)x^\star
- \Big(
C_{\mathcal{I}}^\top (A-BK+\alpha I)x^\star + \alpha d_{\mathcal{I}}
\Big) \\
&= -\alpha\big(C_{\mathcal{I}}^\top x^\star + d_{\mathcal{I}}\big).
\end{align*}
Since $\alpha>0$, it follows that
$C_{\mathcal{I}}^\top x^\star + d_{\mathcal{I}}=0$.
\end{proof}

Prop.~\ref{prop:equilibria_on_faces} shows that any equilibrium associated with a nonempty active set $\mathcal I$ lies on the constraint boundaries indexed by $\mathcal I$. This fact will be used below to analyze motion along active faces and explain why unstable modes need not produce unbounded trajectories.

\begin{corollary}
\label{cor:no_equilibrium_interior_regions}
Let $\mathcal I\neq\emptyset$. If
\(
\overline{\mathcal R}_{\mathcal I}\cap\partial\mathcal C=\emptyset,
\)
then the closed-loop dynamics \eqref{eq:pwa_dynamics} have no equilibrium in $\overline{\mathcal R}_{\mathcal I}$. Consequently, every equilibrium of \eqref{eq:pwa_dynamics} belongs to $\{0\}\cup\partial\mathcal C$.
\end{corollary}

\begin{proof}
By Prop.~\ref{prop:equilibria_on_faces}, any equilibrium associated
with a nonempty active set $\mathcal I$ must satisfy
$C_{\mathcal I}^\top x + d_{\mathcal I}=0$ and therefore lie on
$\partial \mathcal C$. Hence, if
$\overline{\mathcal R}_{\mathcal I}$ does not intersect $\partial\mathcal C$,
it cannot contain a closed-loop equilibrium.
\end{proof}

\subsection{Tangency and Invariant Directions of Active Modes}
\label{subsec:tangency_active_modes}

In this subsection, we show that the affine dynamics associated with a fixed
nonempty active set preserve the corresponding constraint equalities: once
$C_{\mathcal I}^\top x + d_{\mathcal I}=0$ holds, it continues to hold along
trajectories of that mode. As a consequence, any instability of an active mode
can only develop along directions tangent to these constraint boundaries. This
property will be used in the sequel to analyze when unstable modes can lead to
unbounded trajectories in the multi-constraint case.

\begin{proposition}
\label{prop:tangency_active_modes}
Let $\mathcal I\neq\emptyset$, and let $x(\cdot)$ be a trajectory of the closed-loop dynamics \eqref{eq:pwa_dynamics} on an interval $[t_0,t_1]$ such that
\[
x(t)\in\mathcal R_{\mathcal I},\qquad \forall t\in[t_0,t_1].
\]
Then:
\begin{equation}
\frac{d}{dt}\big(C_{\mathcal I}^\top x(t)+d_{\mathcal I}\big)
=
-\alpha\big(C_{\mathcal I}^\top x(t)+d_{\mathcal I}\big),
\qquad \forall t\in[t_0,t_1],
\label{eq:tangency_identity}
\end{equation}
with $C_{\mathcal I}$ and $d_{\mathcal I}$ defined in \eqref{eq:CI_dI_def}. In particular, if
\(
C_{\mathcal I}^\top x(t_0)+d_{\mathcal I}=0,
\)
then
\(
C_{\mathcal I}^\top x(t)+d_{\mathcal I}=0,
\)
for all \(t\in[t_0,t_1].
\)
\end{proposition}

\begin{proof}
For $t\in[t_0,t_1]$, since $x(t)\in\mathcal R_{\mathcal I}$, the closed-loop dynamics \eqref{eq:pwa_dynamics} are given by
\(
\dot x(t)=\tilde A_{\mathcal I}x(t)+\tilde b_{\mathcal I}.
\)
Hence,
\[
\frac{d}{dt}\!\left(C_{\mathcal I}^\top x+d_{\mathcal I}\right)
= C_{\mathcal I}^\top \dot x
= C_{\mathcal I}^\top(\tilde A_{\mathcal I}x+\tilde b_{\mathcal I}).
\]
Using \eqref{eq:Atilde_I_explicit}--\eqref{eq:btilde_I_explicit}, we have
\(
C_{\mathcal I}^\top(\tilde A_{\mathcal I}x+\tilde b_{\mathcal I})
= C_{\mathcal I}^\top(A-BK)x
- C_{\mathcal I}^\top \Pi_{\mathcal I}\!\left(C_{\mathcal I}^\top(A-BK+\alpha I)x+\alpha d_{\mathcal I}\right).
\)
From \eqref{eq:P_I}--\eqref{eq:S_I},
\[
C_{\mathcal I}^\top \Pi_{\mathcal I}
= C_{\mathcal I}^\top BG^{-1}B^\top C_{\mathcal I}\,S_{\mathcal I}^{-1}
= I_{|\mathcal I|}.
\]
Substituting gives
\[
C_{\mathcal I}^\top(\tilde A_{\mathcal I}x+\tilde b_{\mathcal I})
= -\alpha(C_{\mathcal I}^\top x+d_{\mathcal I}),
\]
which yields \eqref{eq:tangency_identity}. The final claim follows by uniqueness of solutions of $\dot z=-\alpha z$ with $z(t_0)=0$.
\end{proof}

Prop.~\ref{prop:tangency_active_modes} shows that the intersection of the
constraint boundaries indexed by $\mathcal I$ is invariant for the affine
mode indexed by $\mathcal I$, as long as the trajectory remains in
$\mathcal R_{\mathcal I}$. We next show that this restricts the directions in
which instability may arise.

Equation~\eqref{eq:tangency_identity} admits the explicit solution
\[
C_{\mathcal I}^\top x(t)+d_{\mathcal I}
= e^{-\alpha(t-t_0)}\big(C_{\mathcal I}^\top x(t_0)+d_{\mathcal I}\big),
\]
for any interval on which the active set $\mathcal I$ remains fixed. In
particular, deviations from the constraint boundaries indexed by
$\mathcal I$ decay exponentially at rate $\alpha$. Thus, while instability
may occur tangentially to the constraint boundaries, the dynamics normal to
those boundaries contract at rate $\alpha$ under a fixed active mode.

\begin{corollary}
\label{cor:tangency_unstable_directions}
Let $\mathcal I\neq\emptyset$. If $v\in\mathbb C^n$ is an eigenvector of
$\tilde A_{\mathcal I}$, defined in \eqref{eq:Atilde_I_explicit}, associated
with an eigenvalue $\lambda\neq-\alpha$, then
\[
C_{\mathcal I}^\top v=0.
\]
In particular, for any eigenvalue satisfying $\Re(\lambda)>0$, the real and
imaginary parts of its eigenvectors are tangent to the constraint boundaries
indexed by $\mathcal I$.
\end{corollary}

\begin{proof}
From \eqref{eq:S_I}, \eqref{eq:P_I}, and \eqref{eq:Atilde_I_explicit},
\[
C_{\mathcal I}^\top \tilde A_{\mathcal I}
=
C_{\mathcal I}^\top A_0
-
C_{\mathcal I}^\top \Pi_{\mathcal I} C_{\mathcal I}^\top (A_0+\alpha I)
=
-\alpha\,C_{\mathcal I}^\top.
\]
Let $v$ satisfy $\tilde A_{\mathcal I}v=\lambda v$. Premultiplying by
$C_{\mathcal I}^\top$ and using the identity above gives
\[
\lambda\, C_{\mathcal I}^\top v
= C_{\mathcal I}^\top \tilde A_{\mathcal I}v
= -\alpha\, C_{\mathcal I}^\top v,
\]
hence $(\lambda+\alpha)C_{\mathcal I}^\top v=0$. If $\lambda\neq -\alpha$, then
$C_{\mathcal I}^\top v=0$.
\end{proof}

Prop.~\ref{prop:tangency_active_modes}, Cor.~\ref{cor:tangency_unstable_directions}, and Prop.~\ref{prop:equilibria_on_faces} show that, for any nonempty active set $\mathcal I$, the corresponding constraint boundaries are invariant under the fixed mode and any unstable eigendirection of $\tilde A_{\mathcal I}$ is tangent to them. Combined with Cor.~\ref{cor:eigs_minphase_multi}, unstable invariant zeros therefore correspond to unstable tangential dynamics. Thus, mode instability alone does not imply unbounded trajectories.

The following example shows that, with multiple constraints, an active-set mode may be
unstable while closed-loop trajectories remain bounded.

\begin{figure}[t]
    \centering
    \includegraphics[width=.8\linewidth]{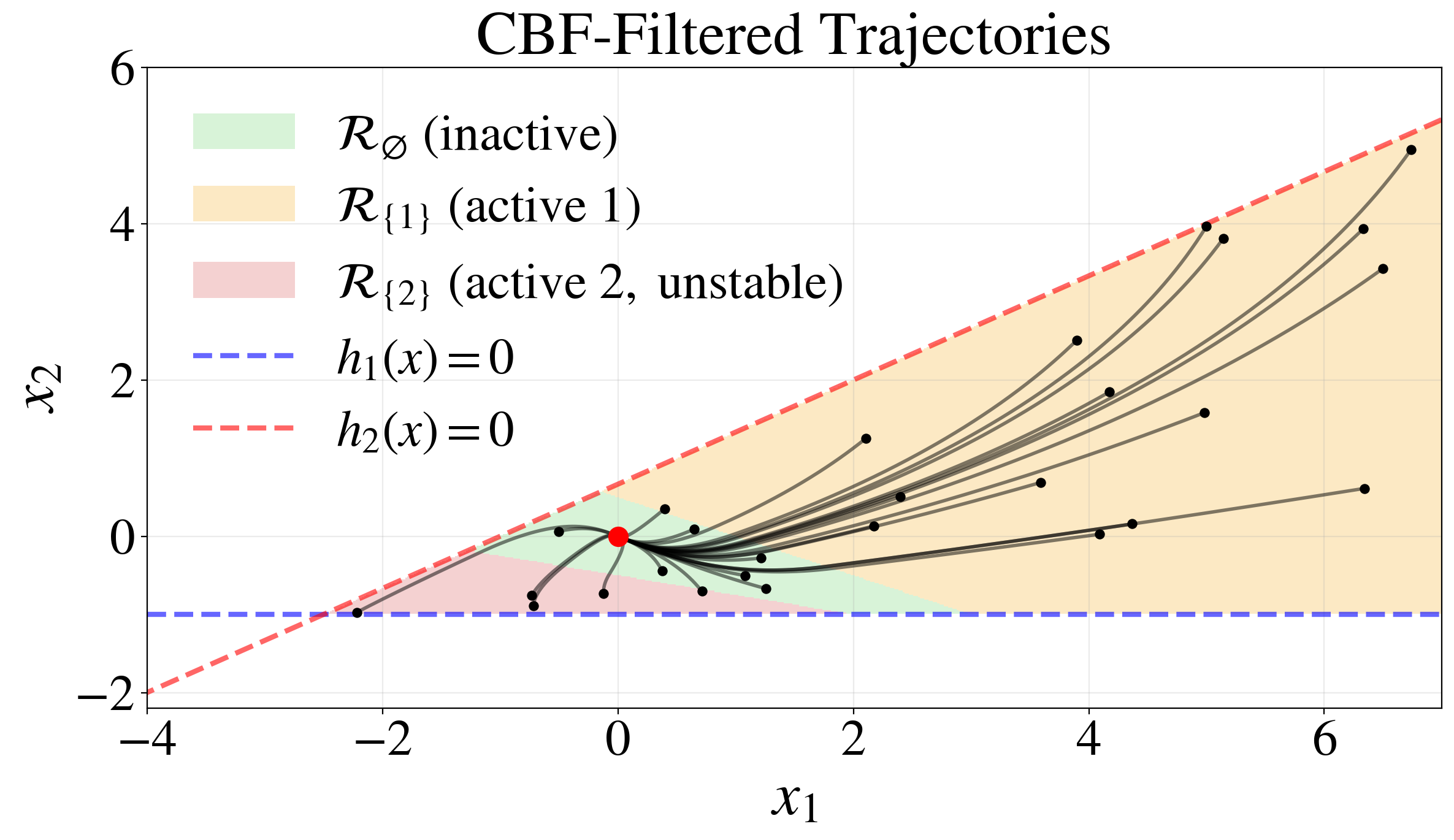}
    \caption{Closed-loop trajectories for
Example~\ref{ex:unstable_no_divergence}.}
    \label{fig:not_unbounded}
\end{figure}

\begin{example}
\label{ex:unstable_no_divergence}
Consider \eqref{eq:system} with $n=2$, $m=1$,
\[
A=\mathrm{diag}(-1,-2),\qquad
B=\begin{bmatrix}1\\[0.3mm]1\end{bmatrix},
\]
and $u_{\rm nom}(x)=-Kx$, where
\(
K=\begin{bmatrix}1&1\end{bmatrix}.
\)
Let $G=1$ and $\alpha=1$ in \eqref{eq:cbf_qp}, and impose the constraints
\[
h_1(x)=x_2+1\ge0,\qquad
h_2(x)=x_1-\tfrac32 x_2+1\ge0.
\]
For $\mathcal I=\{2\}$,
\(
\tilde A_{\{2\}}=\begin{bmatrix}-1&3\\0&1\end{bmatrix},
\)
which has the unstable eigenpair
\(
\lambda=1,\ v=[\,3/2\;\;1\,]^\top
\)
with $c_2^\top v=0$. Thus the unstable direction is tangent to
$h_2(x)=0$. Moreover, $\overline{\mathcal R}_{\{2\}}$ satisfies
$-1\le x_2\le-2/11$, while $\dot x_2=x_2+2\ge1$ in this mode.
Hence trajectories leave this cell in finite time; those in
Fig.~\ref{fig:not_unbounded} switch modes and converge to the origin.
\end{example}

\subsection{Divergence and Recession-Cone Characterization}
\label{subsec:divergence_recession}

Prop.~\ref{prop:tangency_active_modes} and Cor.~\ref{cor:tangency_unstable_directions} show that instability of an active mode $\tilde A_{\mathcal I}$ can only occur along directions tangent to the corresponding constraint boundaries. However, tangential instability alone does not imply an unbounded trajectory. For the persistent fixed-mode mechanism considered below, the unstable direction must remain admissible for arbitrarily large motion along that direction. Here, by divergence we mean $\|x(t)\|\to\infty$.

In the affine-data setting (linear dynamics, affine constraints, quadratic
cost), each active-set region $\mathcal R_{\mathcal I}$ is polyhedral.
Because inactive constraints are typically defined by strict inequalities, we work with the closed cell $\overline{\mathcal R}_{\mathcal I}$, which admits a half-space representation
\begin{equation}
\overline{\mathcal R}_{\mathcal I}
= \{\, x \in \mathbb R^n : H_{\mathcal I} x \le h_{\mathcal I} \,\}.
\label{eq:PI_Hrepr}
\end{equation}
We next make this representation explicit. Let
$\bar{\mathcal I}:=[p]\setminus\mathcal I$, $M:=BG^{-1}B^\top$, and, for
$\mathcal I\neq\emptyset$, define
\[
L_{\mathcal I}:=S_{\mathcal I}^{-1}C_{\mathcal I}^\top(A_0+\alpha I),
\qquad
\ell_{\mathcal I}:=\alpha S_{\mathcal I}^{-1}d_{\mathcal I}.
\]
The KKT multipliers are
$\mu_{\mathcal I}(x)=-(L_{\mathcal I}x+\ell_{\mathcal I})$.
For the inactive constraints, define
\[
\begin{aligned}
R_{\bar{\mathcal I}}
&:=C_{\bar{\mathcal I}}^\top(A_0+\alpha I)
-C_{\bar{\mathcal I}}^\top M C_{\mathcal I}L_{\mathcal I},\\
r_{\bar{\mathcal I}}
&:=\alpha d_{\bar{\mathcal I}}
-C_{\bar{\mathcal I}}^\top M C_{\mathcal I}\ell_{\mathcal I}.
\end{aligned}
\]
Thus, \eqref{eq:PI_Hrepr} is obtained by stacking
\[
-C_{[p]}^\top x\le d_{[p]},\qquad
L_{\mathcal I}x\le-\ell_{\mathcal I},\qquad
-R_{\bar{\mathcal I}}x\le r_{\bar{\mathcal I}},
\]
corresponding respectively to safety, dual feasibility, and the inactive
CBF inequalities. For $\mathcal I=\emptyset$, the multiplier block is absent
and the last inequality becomes
$-C_{[p]}^\top(A_0+\alpha I)x\le\alpha d_{[p]}$.

\begin{definition}[Recession Cone]
For a nonempty closed convex set $\mathcal P$, its recession cone is
\[
\operatorname{rec}(\mathcal P)
:= \{\, v \in \mathbb R^n :
x + \rho v \in \mathcal P,\ \forall x \in \mathcal P,\ \forall \rho \ge 0 \,\}.
\]
Geometrically, $\operatorname{rec}(\mathcal P)$ is the set of directions along
which $\mathcal P$ extends to infinity.
\end{definition}

If
\(
\mathcal P = \{ x : H x \le h \},
\)
then the definition directly gives
\begin{equation}
\operatorname{rec}(\mathcal P)
= \{\, v \in \mathbb R^n : H v \le 0 \,\}.
\label{eq:recHx}
\end{equation}
In particular, if
$\overline{\mathcal R}_{\mathcal I}$ is bounded, then
$\operatorname{rec}(\overline{\mathcal R}_{\mathcal I})=\{0\}$,
so no nonzero ray can remain in the closed cell for arbitrarily large scales.

We now give a sufficient divergence certificate for a fixed active mode: if its equilibrium lies in the closed cell and a real unstable eigenvector is a recession direction, then an unbounded trajectory persists in that cell.

\begin{proposition}
\label{prop:persistent_unstable_ray}
Let $\mathcal I\neq\emptyset$, and let $\overline{\mathcal R}_{\mathcal I}$ denote the closed cell associated with the active set $\mathcal I$. Assume that $\tilde A_{\mathcal I}$, given by \eqref{eq:Atilde_I_explicit}, is invertible, and define
\(
p_{\mathcal I}:=-\tilde A_{\mathcal I}^{-1}\tilde b_{\mathcal I},
\)
where $\tilde b_{\mathcal I}$ is given by \eqref{eq:btilde_I_explicit}. Suppose that
\begin{enumerate}
\item $p_{\mathcal I}\in\overline{\mathcal R}_{\mathcal I}$,
\item $\tilde A_{\mathcal I}$ has a real eigenvalue $\lambda>0$ with real eigenvector $v\neq0$,
\item $v\in\operatorname{rec}(\overline{\mathcal R}_{\mathcal I})$.
\end{enumerate}
Then the closed-loop dynamics \eqref{eq:pwa_dynamics} admit an unbounded trajectory that remains in $\overline{\mathcal R}_{\mathcal I}$ for all $t\ge0$.
\end{proposition}

\begin{proof}
Define \(x(t):=p_{\mathcal I}+e^{\lambda t}v\) for \(t\ge0\). Since \(\tilde A_{\mathcal I}p_{\mathcal I}+\tilde b_{\mathcal I}=0\) and \(\tilde A_{\mathcal I}v=\lambda v\), \(x(\cdot)\) satisfies \(\dot x=\tilde A_{\mathcal I}x+\tilde b_{\mathcal I}\). Moreover, since \(p_{\mathcal I}\in\overline{\mathcal R}_{\mathcal I}\) and \(v\in\operatorname{rec}(\overline{\mathcal R}_{\mathcal I})\), taking \(\rho=e^{\lambda t}\) gives \(x(t)\in\overline{\mathcal R}_{\mathcal I}\) for all \(t\ge0\). By continuity, the closed-loop vector field agrees with this affine mode on the closed cell. Since \(\lambda>0\) and \(v\neq0\), \(\|x(t)\|\to\infty\).
\end{proof}

Prop.~\ref{prop:persistent_unstable_ray} gives a practical sufficient test:
compute a real unstable eigenpair $(\lambda,v)$ of $\tilde A_{\mathcal I}$,
verify $H_{\mathcal I}p_{\mathcal I}\le h_{\mathcal I}$, and check
$H_{\mathcal I}v\le0$. When these conditions hold, the corresponding
unstable ray remains in $\overline{\mathcal R}_{\mathcal I}$ and produces an
unbounded trajectory.

Failure of this certificate does not establish boundedness, even for a fixed
mode, and Prop.~\ref{prop:persistent_unstable_ray} does not characterize
divergence involving switching. The global Lyapunov conditions in
Sec.~\ref{sec:boundedness} rule out such behavior when feasible, whereas
their infeasibility is inconclusive.

\begin{example}
\label{ex:strip_lp_divergence}
Consider \eqref{eq:system} with
\[
A=\begin{bmatrix}2&0\\[0.3mm]1&0\end{bmatrix},
\qquad
B=\begin{bmatrix}1\\[0.3mm]1\end{bmatrix},
\qquad
K=\begin{bmatrix}4&-1\end{bmatrix},
\]
with nominal controller \(u_{\rm nom}(x)=-Kx\), and let \(\alpha=1\) and
\(G=1\) in \eqref{eq:cbf_qp}. Impose the constraints
\[
h_1(x)=x_2+1\ge0,
\qquad
h_2(x)=-x_2+1\ge0.
\]

For \(\mathcal I=\{1\}\) and \(\mathcal I=\{2\}\), the corresponding active
modes are unstable, with unstable eigenspace
\(\operatorname{span}\{e_1\}\) associated with \(\lambda=1\). Their mode
equilibria are
\(
p_{\{1\}}=\begin{bmatrix}0&-1\end{bmatrix}^\top
\)
and
\(
p_{\{2\}}=\begin{bmatrix}0&1\end{bmatrix}^\top,
\)
and belong to their respective closed cells.

The recession test \eqref{eq:recHx} verifies \(v=e_1\) for
\(\mathcal I=\{1\}\) and \(v=-e_1\) for \(\mathcal I=\{2\}\).
Hence Prop.~\ref{prop:persistent_unstable_ray} applies, yielding
\[
x_{\{1\}}(t)
=\begin{bmatrix}e^t\\[-0.5mm]-1\end{bmatrix},
\qquad
x_{\{2\}}(t)
=\begin{bmatrix}-e^t\\[0.3mm]1\end{bmatrix}.
\]
As shown in Fig.~\ref{fig:unbounded}, these trajectories remain in their
respective closed cells and diverge along the boundaries \(x_2=-1\) and
\(x_2=1\).
\end{example}

\begin{figure}[t]
    \centering
    \includegraphics[width=\linewidth]{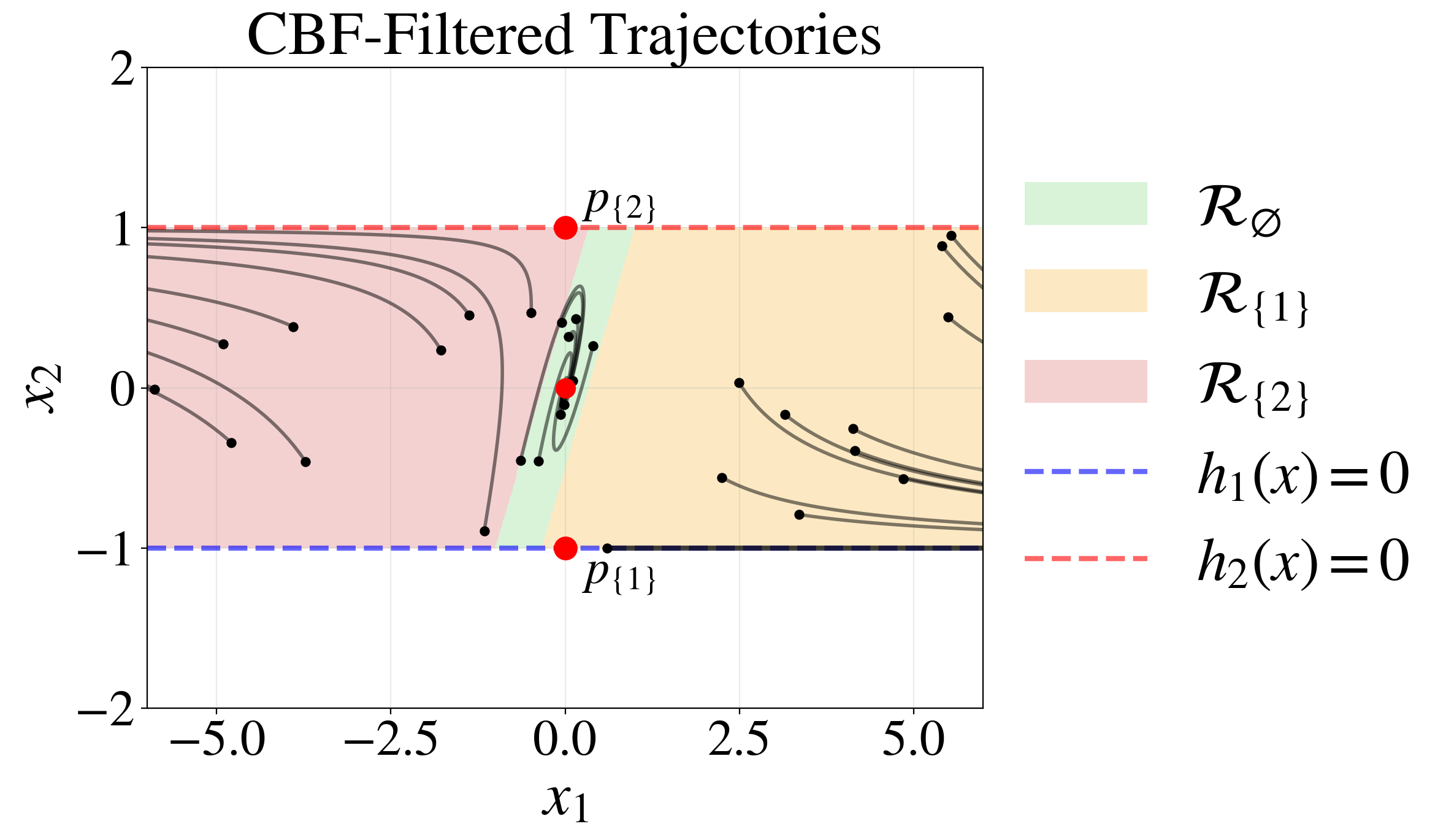}
    \caption{Closed-loop trajectories for
    Example~\ref{ex:strip_lp_divergence}.}
    \label{fig:unbounded}
\end{figure}

\section{Boundedness and Invariant-Set Behavior}
\label{sec:boundedness}

The previous section showed that instability of an active mode need not imply unbounded behavior. Prop.~\ref{prop:persistent_unstable_ray} identified a sufficient fixed-mode mechanism for unbounded behavior: an unstable direction that remains in the recession cone of $\overline{\mathcal R}_{\mathcal I}$. When this mechanism is absent, trajectories may switch among modes, and boundedness cannot be inferred from individual mode stability alone.

This section formalizes these observations using Lyapunov and invariance tools that account for such switching. We first give a strict common quadratic Lyapunov condition for global exponential stability, and then introduce region-wise Lyapunov conditions that exploit the polyhedral active-set partition. These results provide global stability and boundedness certificates when a strict common quadratic certificate is unavailable.

\subsection{Global Exponential Stability via a CQLF}
\label{subsec:cqlf_ges}

We now present a sufficient condition for global exponential stability (GES) of the origin based on the existence of a common quadratic Lyapunov function (CQLF) for all active-set modes of the closed-loop dynamics \eqref{eq:pwa_dynamics}. Specifically, we seek a common quadratic metric \(P\succ0\) under which the linear parts of all modes are strictly contracting. Let
\[
\mathcal J
:=
\{\mathcal I \subseteq [p] :
\mathcal R_{\mathcal I} \neq \emptyset \}
\]
denote the set of active sets that occur on \(\mathcal C\).

\begin{theorem}
\label{thm:ges_cqlf}
Suppose there exist \(P\succ0\) and \(\epsilon>0\) such that
\begin{equation}
P\tilde A_{\mathcal I}+\tilde A_{\mathcal I}^\top P\preceq-\epsilon I,
\qquad \forall \mathcal I\in\mathcal J,
\label{eq:cqlf_strict}
\end{equation}
where \(\tilde A_{\mathcal I}\) is given by
\eqref{eq:Atilde_I_explicit} for \(\mathcal I\neq\emptyset\) and by
\eqref{eq:nominal_mode} for \(\mathcal I=\emptyset\). Then the origin is GES in \(\mathcal C\) for the closed-loop dynamics \eqref{eq:pwa_dynamics}.
\end{theorem}

\begin{proof}
Under the standing assumptions, the closed-loop dynamics \eqref{eq:pwa_dynamics} define a continuous piecewise-affine vector field on the convex forward-invariant set $\mathcal C$, with affine dynamics $f(x)=\tilde A_{\mathcal I}x+\tilde b_{\mathcal I}$ on each region $\mathcal R_{\mathcal I}$. Therefore, the hypotheses of \cite[Thm.~1]{pavlov2007convergence} apply on $\mathcal C$. Since \eqref{eq:cqlf_strict} holds for all $\mathcal I\in\mathcal J$, that result yields exponential convergence between any two solutions in the quadratic metric induced by $P$. Since $u^\star(0)=0$, the origin is an equilibrium; taking one solution to be identically zero yields GES of the origin in $\mathcal C$.
\end{proof}

\begin{corollary}
\label{cor:unique_equilibrium_cqlf}
Under the assumptions of Thm.~\ref{thm:ges_cqlf}, the origin is the unique equilibrium of the closed-loop dynamics \eqref{eq:pwa_dynamics} in \(\mathcal C\).
\end{corollary}

\begin{proof}
Any nonzero equilibrium would generate a constant nonzero trajectory, contradicting global exponential convergence to the origin established in Thm.~\ref{thm:ges_cqlf}.
\end{proof}

\begin{remark}
\label{rem:lmi_verification}
The LMIs in \eqref{eq:cqlf_strict} provide a tractable sufficient test for GES
on $\mathcal C$, but the condition is conservative: even if every
$\tilde A_{\mathcal I}$ is Hurwitz, a CQLF may not exist.
\end{remark}

\begin{remark}
If the strict condition in \eqref{eq:cqlf_strict} is replaced by
\[
P\tilde A_{\mathcal I}+\tilde A_{\mathcal I}^\top P \preceq 0,
\qquad \forall \mathcal I\in\mathcal J,
\]
then, under the standing assumptions that the closed-loop dynamics \eqref{eq:pwa_dynamics} are continuous piecewise affine on the convex forward-invariant set \(\mathcal C\) and that the origin is an equilibrium, the line-segment argument of \cite[Lemma~2]{pavlov2005convergent} with zero contraction rate implies \(\dot V(x)\le0\) on \(\mathcal C\) for \(V(x)=x^\top Px\). Hence Thm.~\ref{thm:boundedness_lasalle} yields boundedness and convergence to the largest invariant subset of
\(
\{x\in\mathcal C:\dot V(x)=0\}.
\)
\end{remark}

\subsection{Region-Wise Lyapunov Certificates}
\label{subsec:boundedness_lasalle}

We now present a condition weaker than strict CQLF-based GES.
Instead of requiring strict contraction of all mode matrices, we require a common
quadratic function to be nonincreasing along the full affine dynamics on each region.

\begin{theorem}
\label{thm:boundedness_lasalle}
Let \(V(x):=x^\top Px\) with \(P\succ0\). If
\begin{equation}
\nabla V(x)^\top\big(\tilde A_{\mathcal I}x+\tilde b_{\mathcal I}\big)\le0,
\qquad
\forall x\in\overline{\mathcal R}_{\mathcal I},
\quad
\forall \mathcal I\in\mathcal J,
\label{eq:lasalle_condition}
\end{equation}
then every trajectory of \eqref{eq:pwa_dynamics} starting in \(\mathcal C\) is bounded and converges to the largest invariant subset of
\begin{equation}
\mathcal E_V:=\{x\in\mathcal C:\dot V(x)=0\}.
\label{eq:EV_set}
\end{equation}
\end{theorem}

\begin{proof}
Since \(\mathcal C\) is forward invariant and \(V(x)=x^\top Px\) with \(P\succ0\),
\eqref{eq:lasalle_condition} implies that \(V(x(t))\) is nonincreasing along trajectories of \eqref{eq:pwa_dynamics}. Hence every trajectory starting in \(\mathcal C\) remains in a bounded sublevel set of \(V\), and is therefore bounded. Convergence to the largest invariant subset of $\mathcal E_V$ in \eqref{eq:EV_set} then follows from LaSalle's invariance principle for continuous piecewise-affine systems; see, e.g., \cite{branicky1998multiple,johansson2003piecewise}.
\end{proof}

\begin{remark}
\label{rem:lasalle_vs_cqlf}
Thm.~\ref{thm:boundedness_lasalle} does not require each matrix \(\tilde A_{\mathcal I}\) to be Hurwitz, and can therefore certify boundedness even when some active-set modes are unstable. Moreover, if the largest invariant subset of $\mathcal E_V$ in \eqref{eq:EV_set} is \(\{0\}\), then LaSalle's invariance principle yields GAS of the origin in \(\mathcal C\), and hence uniqueness of the equilibrium in \(\mathcal C\).

This condition should be contrasted with the strict CQLF requirement \eqref{eq:cqlf_strict}, which enforces strict contraction of the linear parts uniformly across modes. By contrast, \eqref{eq:lasalle_condition} only requires Lyapunov decrease on each closed cell \(\overline{\mathcal R}_{\mathcal I}\) for the full affine vector field \(\tilde A_{\mathcal I}x+\tilde b_{\mathcal I}\). These geometric restrictions can exclude directions responsible for CQLF infeasibility.
\end{remark}

Under the hypotheses of Thm.~\ref{thm:boundedness_lasalle}, every trajectory starting in \(\mathcal C\) is bounded. Moreover, if the active set is constant and nonempty on some interval \([T,\infty)\), i.e., \(x(t)\in\mathcal R_{\mathcal I}\) for all \(t\ge T\) with \(\mathcal I\neq\emptyset\), then Prop.~\ref{prop:tangency_active_modes} yields
\[
C_{\mathcal I}^\top x(t)+d_{\mathcal I}
=
e^{-\alpha(t-T)}\big(C_{\mathcal I}^\top x(T)+d_{\mathcal I}\big)\to0.
\]
Thus, a persistent nonempty mode drives the trajectory toward the corresponding constraint face, while boundedness rules out divergence.

The next corollary gives a convenient, though conservative, matrix-only condition implying \eqref{eq:lasalle_condition}.

\begin{corollary}
\label{cor:matrix_to_lasalle}
Let \(V(x):=x^\top Px\) with \(P\succ0\). If
\[
P\tilde A_{\mathcal I}+\tilde A_{\mathcal I}^\top P\preceq0,
\qquad \forall \mathcal I\in\mathcal J,
\]
then \(\dot V(x)\le0\) for all \(x\in\mathcal C\). Equivalently, \eqref{eq:lasalle_condition} holds.
\end{corollary}

\begin{proof}
Let \(f\) denote the continuous piecewise-affine closed-loop vector field, and fix \(x_1,x_2\in\mathcal C\). Since \(\mathcal C\) is convex, the segment
\(\gamma(s)=x_2+s(x_1-x_2)\), \(s\in[0,1]\), lies in \(\mathcal C\) and intersects only finitely many active-set regions. On each such segment,
\[
\frac{d}{ds}f(\gamma(s))
=
\tilde A_{\mathcal I}(x_1-x_2).
\]
Because \(f\) is continuous across switching boundaries, summing these increments introduces no jump terms; in particular, the affine offsets cancel within each region. Hence, with \(d:=x_1-x_2\),
\[
d^\top P\big(f(x_1)-f(x_2)\big)
\!=\!
\frac12\sum_k \Delta s_k\,
d^\top\!\left(P\tilde A_{\mathcal I_k}
\!+\!\tilde A_{\mathcal I_k}^\top P\right)d
\le\!0,
\]
where the intervals of length \(\Delta s_k\) correspond to the regions crossed by the segment; cf.~\cite[Lemma~2]{pavlov2005convergent}. Taking \(x_2=0\) and using \(f(0)=0\) gives \(x^\top Pf(x)\le0\), and therefore
\[
\dot V(x)=2x^\top Pf(x)\le0,
\qquad \forall x\in\mathcal C,
\]
which is \eqref{eq:lasalle_condition}.
\end{proof}

The condition \eqref{eq:lasalle_condition} can also be enforced directly on each region via
a tractable LMI relaxation that exploits the polyhedral description of $\overline{\mathcal R}_{\mathcal I}$.

For \(P=P^\top\in\mathbb R^{n\times n}\), define
\begin{equation}
Q_{\mathcal I}(P):=\tilde A_{\mathcal I}^\top P+P\tilde A_{\mathcal I},
\qquad
q_{\mathcal I}(P):=P\tilde b_{\mathcal I}.
\label{eq:Qq_def}
\end{equation}

\begin{proposition}
\label{prop:lmi_lasalle}
Let \(V(x):=x^\top Px\) with \(P=P^\top\), and suppose \(\overline{\mathcal R}_{\mathcal I}\) is represented as in \eqref{eq:PI_Hrepr}. Let \(Q_{\mathcal I}(P)\) and \(q_{\mathcal I}(P)\) be defined by \eqref{eq:Qq_def}. If there exist \(P\succ0\) and \(\lambda_{\mathcal I}\ge0\) such that
\begin{equation}
\begin{bmatrix}
Q_{\mathcal I}(P) &
q_{\mathcal I}(P)-\tfrac12 H_{\mathcal I}^\top\lambda_{\mathcal I}\\[1mm]
\star &
\lambda_{\mathcal I}^\top h_{\mathcal I}
\end{bmatrix}
\preceq0,
\label{eq:LMI_lasalle_region}
\end{equation}
then \(\dot V(x)\le0\) for all \(x\in\overline{\mathcal R}_{\mathcal I}\). Consequently, if \eqref{eq:LMI_lasalle_region} holds for every \(\mathcal I\in\mathcal J\) with the same \(P\succ0\), then \eqref{eq:lasalle_condition} holds on \(\mathcal C\).

Moreover, if there exist \(P\succ0\), \(\varepsilon>0\), and multipliers \(\{\lambda_{\mathcal I}\ge0\}_{\mathcal I\in\mathcal J}\) such that
\begin{equation}
\begin{bmatrix}
Q_{\mathcal I}(P)+\varepsilon I &
q_{\mathcal I}(P)-\tfrac12 H_{\mathcal I}^\top\lambda_{\mathcal I}\\[1mm]
\star &
\lambda_{\mathcal I}^\top h_{\mathcal I}
\end{bmatrix}
\preceq0,
\qquad \forall \mathcal I\in\mathcal J,
\label{eq:strictlmi}
\end{equation}
then \(\dot V(x)\le-\varepsilon\|x\|^2\) on \(\mathcal C\), and the origin is GES in \(\mathcal C\) for the closed-loop dynamics \eqref{eq:pwa_dynamics}.
\end{proposition}

\begin{proof}
Fix \(\mathcal I\in\mathcal J\), and define
\[
M_{\mathcal I}:=
\begin{bmatrix}
Q_{\mathcal I}(P) &
q_{\mathcal I}(P)-\tfrac12 H_{\mathcal I}^\top\lambda_{\mathcal I}\\
\star & \lambda_{\mathcal I}^\top h_{\mathcal I}
\end{bmatrix}.
\]
If \(M_{\mathcal I}\preceq0\), then for every \(x\in\mathbb R^n\),
\[
\begin{bmatrix}x\\1\end{bmatrix}^\top\!
M_{\mathcal I}
\begin{bmatrix}x\\1\end{bmatrix}
\!=\!
x^\top Q_{\mathcal I}(P)x
+2q_{\mathcal I}(P)^\top x
+\lambda_{\mathcal I}^\top(h_{\mathcal I}-H_{\mathcal I}x)
\!\le\!0.
\]
For \(x\in\overline{\mathcal R}_{\mathcal I}\), we have
\(h_{\mathcal I}-H_{\mathcal I}x\ge0\). Since
\(\lambda_{\mathcal I}\ge0\),
\[
x^\top Q_{\mathcal I}(P)x+2q_{\mathcal I}(P)^\top x\le0,
\qquad
\forall x\in\overline{\mathcal R}_{\mathcal I}.
\]
Using \eqref{eq:Qq_def}, this is exactly \(\dot V(x)\le0\) on
\(\overline{\mathcal R}_{\mathcal I}\). If
\eqref{eq:LMI_lasalle_region} holds for every \(\mathcal I\in\mathcal J\),
then \(\dot V(x)\le0\) on \(\mathcal C\).

For the exponential claim, replacing \(Q_{\mathcal I}(P)\) by
\(Q_{\mathcal I}(P)+\varepsilon I\) gives
\(\dot V(x)\le-\varepsilon\|x\|^2\) on \(\mathcal C\). Since
\[
\lambda_{\min}(P)\|x\|^2\le V(x)\le\lambda_{\max}(P)\|x\|^2,
\]
we obtain
\[
\dot V(x)\le-(\varepsilon/\lambda_{\max}(P))V(x),
\]
and hence GES of the origin in \(\mathcal C\).
\end{proof}

Condition \eqref{eq:LMI_lasalle_region} is a sufficient multiplier-based
relaxation of the exact region-wise inequality \eqref{eq:lasalle_condition};
thus, infeasibility does not imply that \eqref{eq:lasalle_condition} fails.
Since \(Q_{\mathcal I}(P)\) is a principal block of the LMI,
\eqref{eq:LMI_lasalle_region} necessarily requires
\(Q_{\mathcal I}(P)\preceq0\); hence, unlike the exact region-wise condition,
this relaxation cannot accommodate mode matrices with right-half-plane
eigenvalues. Likewise, \eqref{eq:strictlmi} implies the strict CQLF condition
\eqref{eq:cqlf_strict}. The test requires one LMI per realized active set;
under LICQ, at most
\(\sum_{k=0}^{\min\{m,p\}}\binom{p}{k}\)
candidate active sets can occur.

\section{Illustrative Examples and Discussion}

We begin with an example in which the global CQLF test
\eqref{eq:cqlf_strict} succeeds, yielding GES via
Thm.~\ref{thm:ges_cqlf}.

\begin{example}
\label{ex:cqlf_works_nondiagonal_nonsymmetric}
Consider \eqref{eq:system} with \(K=0\), \(G=I_2\), and \(\alpha=1\), where
\begin{equation*}
A=
\begin{bmatrix}
-1.5421 & 1.5125\\
-0.2300 & -0.8362
\end{bmatrix},
\qquad
B=I_2.
\end{equation*}
Let the constraints be \(h_1(x)=x_1+1\ge0\) and \(h_2(x)=x_2+1\ge0\), so that
\[
\mathcal C=\{x\in\mathbb R^2:\ x_1\ge -1,\ x_2\ge -1\}
\]
is unbounded. The resulting safety filter induces a continuous PWA closed-loop system with four active-set modes, whose linear parts are
\begin{equation}
\begin{aligned}
\tilde A_{\emptyset} &= A,
&
\tilde A_{\{1\}} &=
\begin{bmatrix}
-1 & 0\\
-0.2300 & -0.8362
\end{bmatrix},
\\[1ex]
\tilde A_{\{1,2\}} &=-I_2, &
\tilde A_{\{2\}} &=
\begin{bmatrix}
-1.5421 & 1.5125\\
0 & -1
\end{bmatrix}.
\end{aligned}
\end{equation}
The strict CQLF LMIs in \eqref{eq:cqlf_strict} hold with
\(P=I_2\) and \(\epsilon=0.9\). It then follows from
Thm.~\ref{thm:ges_cqlf} that the origin is GES in \(\mathcal C\), as
illustrated in Fig.~\ref{fig:cqlf}.
\end{example}

The next example shows that failure of the strict CQLF condition does not
preclude GAS: a non-strict Lyapunov certificate together with invariant-set
analysis can still establish convergence to the origin.

\begin{example}
\label{ex:lasalle_not_cqlf_concise}
Consider \eqref{eq:system} with \(K=0\), \(G=I_2\), and \(\alpha=1\), where
\[
A=
\begin{bmatrix}
-1 & 2\\
-2 & -1
\end{bmatrix},
\qquad
B=I_2,
\]
and constraints
\(
h_1(x)=x_1+1\ge0,
\;
h_2(x)=x_2+1\ge0.
\)
The safety filter induces a continuous piecewise-affine closed loop with four active-set modes.

The strict CQLF condition \eqref{eq:cqlf_strict} in Thm.~\ref{thm:ges_cqlf} is infeasible. Indeed, imposing it on the modes \(\mathcal I=\{1\}\) and \(\mathcal I=\{2\}\) yields the contradictory conditions \(c>a\) and \(a>c\) for
\[
P=
\begin{bmatrix}
a & b\\
b & c
\end{bmatrix}\succ0.
\]

By contrast, \eqref{eq:LMI_lasalle_region} is feasible with \(P=I_2\) and,
using the safety--dual--inactive row ordering,
\[
\lambda_{\emptyset}=0,\:
\lambda_{\{1\}}=\lambda_{\{2\}}=(0,0,1,0)^\top,\:
\lambda_{\{1,2\}}=(0,0,0,1)^\top.
\]
Moreover, direct substitution gives \(\dot V=-2\|x\|^2\) on the nominal
cell, \(\dot V\le-1/2\) on the singleton cells, and \(\dot V\le-1\) on the
two-constraint cell. Thus \(\mathcal E_V=\{0\}\), and
Thm.~\ref{thm:boundedness_lasalle} yields GAS in \(\mathcal C\), as
illustrated in Fig.~\ref{fig:lassal}.
\end{example}

\begin{figure}[t]
    \centering
    \includegraphics[width=.9\linewidth]{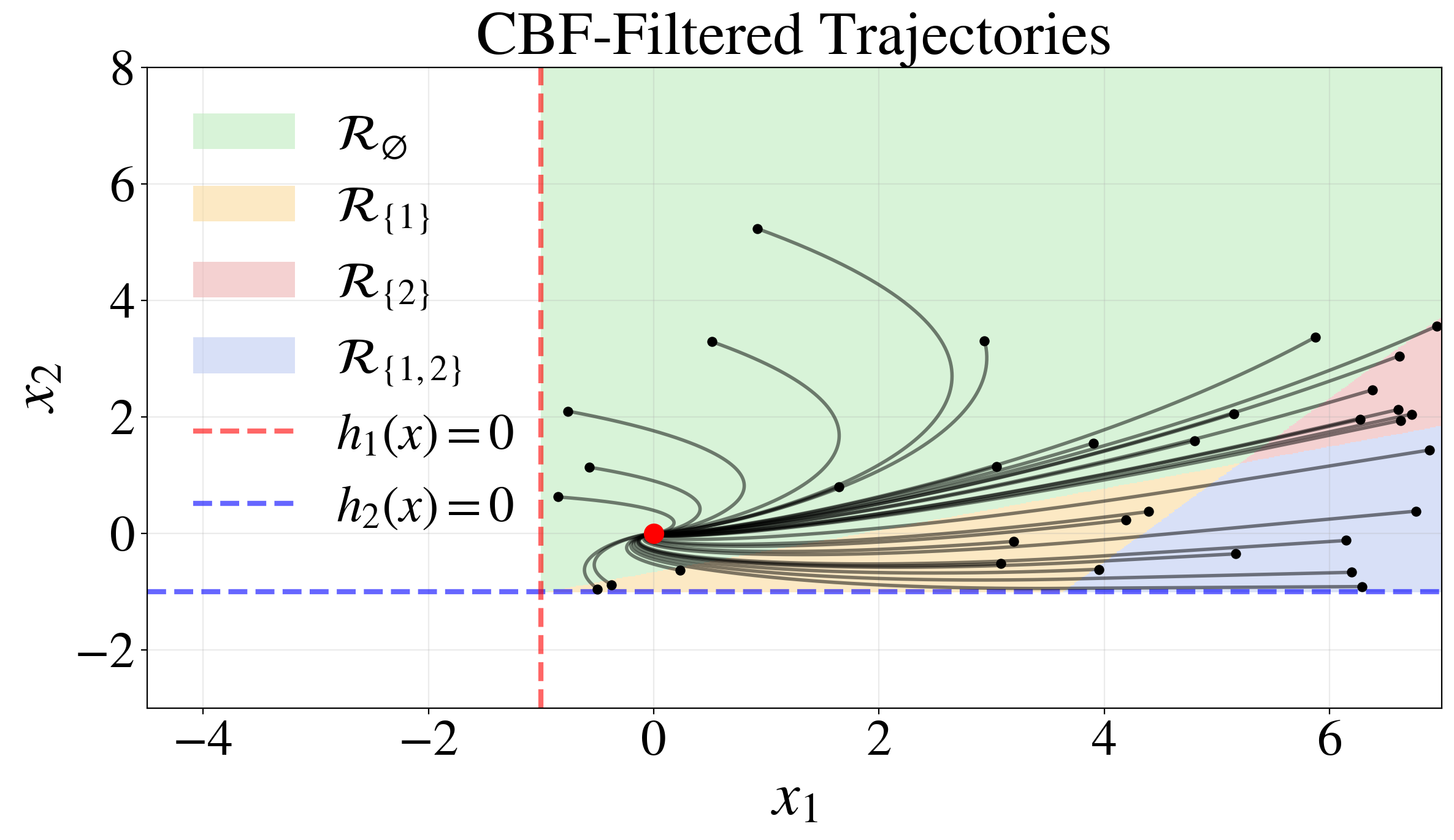}
    \caption{Closed-loop trajectories for
Example~\ref{ex:cqlf_works_nondiagonal_nonsymmetric}.}
    \label{fig:cqlf}
\end{figure}

\begin{figure}[t]
    \centering
    \includegraphics[width=.9\linewidth]{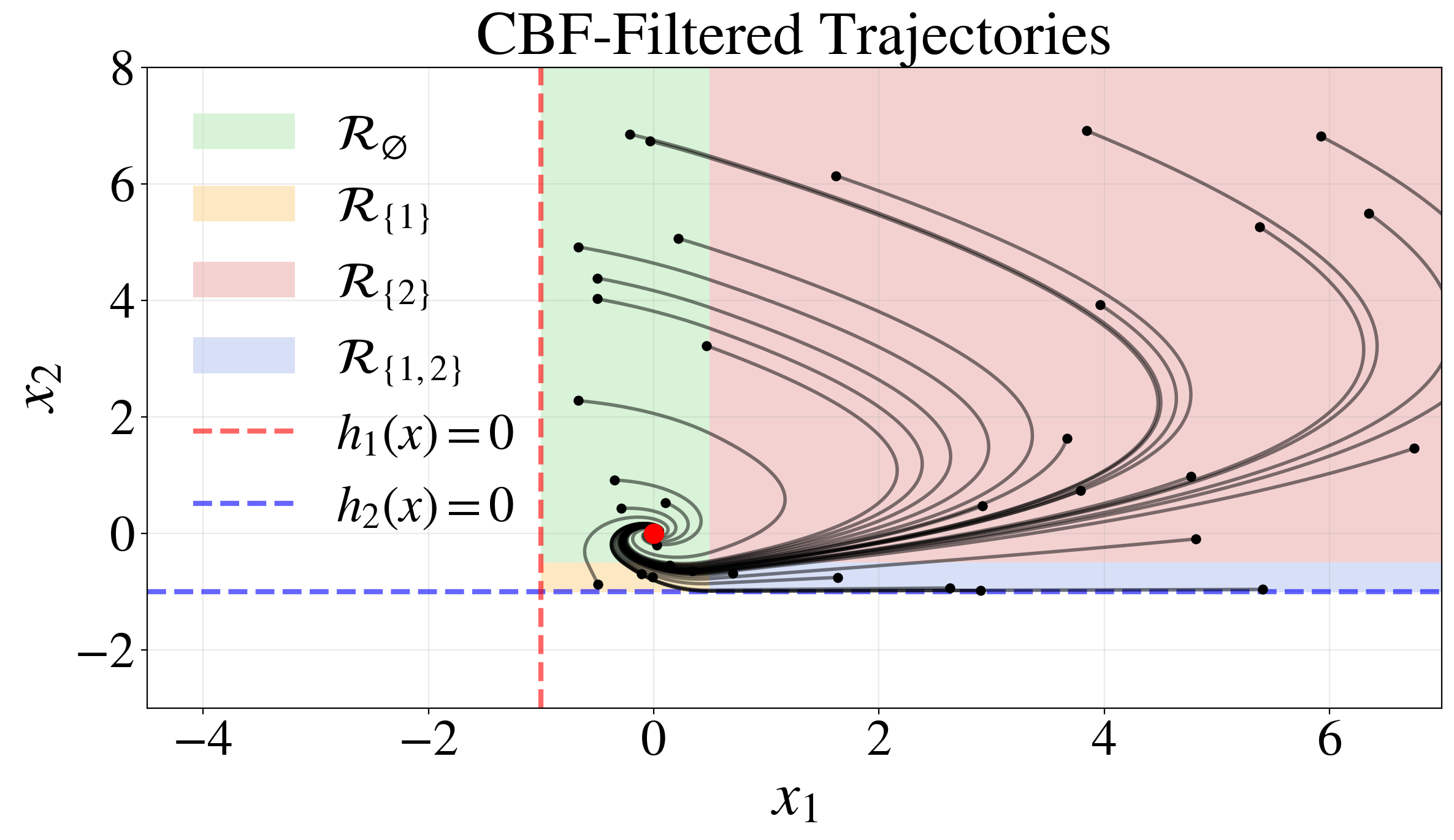}
    \caption{Closed-loop trajectories for
Example~\ref{ex:lasalle_not_cqlf_concise}.}
    \label{fig:lassal}
\end{figure}

\section{Conclusions}

This paper studied multi-constraint CBF-QP safety filters for linear systems with affine state constraints. We showed that nonempty active sets can create boundary equilibria, that under a fixed active set instability can only develop tangentially to the active faces, and provided a sufficient recession-cone condition for divergence under a fixed active set. We also derived a minimum-phase test for mode stability and Lyapunov/LMI-based certificates for global exponential stability, boundedness, and asymptotic convergence. Future work will focus on less conservative certificates based on piecewise quadratic Lyapunov functions adapted to the polyhedral active-set structure.

\bibliography{references}

\bibliographystyle{IEEEtran}

\end{document}